\documentclass[10pt, reqno]{amsart}
\usepackage{amssymb, amscd, mathrsfs, wasysym}
\usepackage[mathcal]{euscript}
\usepackage{verbatim}
\usepackage{a4wide}

\newtheorem{theorem}{Theorem}[section]
\newtheorem{corollary}[theorem]{Corollary}
\newtheorem{proposition}[theorem]{Proposition}
\newtheorem{lemma}[theorem]{Lemma}
\newtheorem{definition}[theorem]{Definition}
\newtheorem{example}[theorem]{Example}
\theoremstyle{remark}
\newtheorem{remark}[theorem]{Remark}

\numberwithin{equation}{section}
\newcommand \pibf {{\boldsymbol \pi}}
\newcommand \omegabf {{\boldsymbol \omega}}
\newcommand \alphabf {{\boldsymbol \alpha}}

\newcommand \barK {\underline K}
\newcommand \Kbar {\overline K}

\newcommand\IPT[2]{\ensuremath{\langle #1, #2 \rangle_{\T}}}
\DeclareMathOperator{\injec}{inj}

\newcommand \Cut {\text{Cut}}
\newcommand \ben 	{\begin{enumerate}}
\newcommand \een 	{\end{enumerate}}

\newcommand \la 		\langle
\newcommand \ra 		\rangle
\newcommand \gammad {{\dot \gamma}}
\newcommand \be 		{\begin{equation}}
\newcommand \ee 		{\end{equation}}
\newcommand \bel[1] 	{\begin{equation}\label{#1}}
\newcommand \R 		{\mathbb{R}}
\newcommand \eps 	{\epsilon}
\newcommand \del 	\partial

\newcommand \NullConj 	{\text {Null Conj}}

\newcommand \NullInj 	{\text {Null Inj}}

\newcommand \Ncal	{\mathcal N}

\newcommand \Fcal	{\mathcal F}

\newcommand \Hcal 	{{\mathcal H}}

\newcommand \Inj 	{\text {Inj}}

\newcommand \M		{M}

\newcommand \RR 		{\mathbb R}
\newcommand \Rm 		{{\textsl Rm}}

\newcommand \J		Y

\newcommand	\Phid	{\dot \Phi}

\newcommand	\g    	g
\newcommand	\gT   	{{g_T}}
\newcommand	\T		{\ensuremath{T}}

\newcommand	\norm[1]{\ensuremath{\big| #1 \big|}}
\newcommand	\normT[1]{\ensuremath{| #1 |_T}}

\DeclareMathOperator{\grad}{grad}

\let\oldmarginpar\marginpar
\renewcommand\marginpar[1]{\-\oldmarginpar[\raggedleft\footnotesize #1]%
{\raggedright\footnotesize #1}}

\begin{document}

\title[Null injectivity estimate under an upper bound on the curvature]
{Null injectivity estimate
\\
under an upper bound on the curvature}
\author[James D.E. Grant]{James~D.E.\ Grant}
\address{Institut f\"{u}r Grundlagen der Bauingenieurwissenschaften\\
Leopold-Franzens-Univer\-sit\"{a}t Innsbruck\\
Technikerstrasse 13\\ 6020 Innsbruck\\ Austria}
\email{James.Grant@univie.ac.at}
\author[Philippe G. L{\tiny e}Floch]{Philippe~G.\ L{\smaller e}Floch}
\address{Laboratoire Jacques-Louis Lions \& Centre National de la Recherche Scientifique \\
Universit\'e Pierre et Marie Curie (Paris 6) \\ 4 Place Jussieu \\ 75252 Paris \\ France.}
\email{pgLeFloch@gmail.com.}
%
\date{}

\begin{abstract}
We establish a uniform estimate for the injectivity radius of the past null cone of a point
in a general Lorentzian manifold foliated by spacelike hypersurfaces and
satisfying an upper curvature bound. Precisely, our main assumptions are, on one hand,
upper bounds on the null curvature of the spacetime and the lapse function of the foliation,
and sup-norm bounds on the deformation tensors of the foliation.
Our proof is inspired by techniques from Riemannian geometry, and it should be noted that
we impose no restriction on the size of the curvature or deformation tensors,
and allow for metrics that are ``far'' from the Minkowski one.
The relevance of our estimate is illustrated with a class of plane-symmetric spacetimes
which satisfy our assumptions but admit no uniform lower bound on the curvature not even in the $L^2$ norm.
The conditions we put forward, therefore, lead to a uniform control of the spacetime geometry
and should be useful in the context of general relativity.
\end{abstract}

\maketitle


\section{Introduction}

In this paper, we consider time-oriented, Lorentzian manifolds
satisfying certain geometric bounds and, by suitably adapting techniques from Riemannian geometry,
we derive geometric estimates about null cones, that is, the boundary of the past of a point in the manifold.
Our main purpose is to investigate the role of a one-sided bound on the curvature,
as opposed to two-sided or integral bounds and, specifically,
to establish a uniform lower bound on the injectivity radius of null cones.

In the recent work~\cite{KR4}, Klainerman and Rodnianski derived such an estimate for the null
injectivity radius of a four-dimensional Ricci-flat Lorentzian manifold in terms of the $L^2$ norm of the curvature tensor
on spacelike hypersurfaces and additional geometric quantities. Null cones play a
central role in the (harmonic) analysis of nonlinear wave equations
and having a good control of null cones allows one, for instance, to construct parametrices and tackle
the initial value problem, as explained in~\cite{KR4,KR:breakdown}.
Recall also that Chen and LeFloch~\cite{ChenLeFloch,LeFloch} covered Lorentzian manifolds
whose Riemann curvature is bounded above and below.
On the other hand, imposing solely an upper curvature bound
 raises new conceptual and technical difficulties, overcome in the present paper.

Our aim is thus to identify minimal conditions required to obtain an injectivity estimate,
without a~priori imposing the Einstein equations. This is important if one wants to cover large classes of spacetimes
which need not be vacuum and, even in the vacuum, this investigation should contribute to identify optimal conditions.

We introduce here a new technique of proof that uses only a one-sided curvature bound
and relies entirely on differential geometric arguments. In
particular, we avoid assumptions concerning the existence of coordinate
systems in which the metric would be close to the flat metric.
We state the assumptions required for a null injectivity estimate directly in terms of geometrical data,
especially an upper bound on the null curvature and a bound on deformation tensors and lapse function.
In turn, this places
our results more in line with standard injectivity radius estimates in
Riemannian geometry, such as those of Cheeger~\cite{Cheeger:finiteness},
Heintze and Karcher~\cite{HeintzeKarcher}, and
Cheeger, Gromov, and Taylor~\cite{CheegerGromovTaylor}.

The present paper builds on an extensive literature in both the Riemannian and the Lorentzian settings,
and we especially gained insights from the papers by Ehrlich
and co-authors~\cite{BE1,BE2,BEE,EhrlichSanchez}. Recall that sectional curvature bounds in the context of
Lorentzian geometry were studied by Andersson and Howard~\cite{AnderssonHoward}
and allowed them to derive various comparison and rigidity theorems.
More recently, Alexander and Bishop~\cite{AB} have derived triangle comparison theorems
for semi-Riemannian manifolds satisfying sectional curvature bounds.

Furthermore, the results in the present paper are relevant to, and provide a set-up for analyzing, the long-time behavior
of solutions to Einstein's field equations of general relativity.
Our presentation is directly applicable to identify
all geometric information required {\sl before} imposing the Einstein equations.

This paper is organized as follows.
In Section~\ref{obs}, we begin with some terminology
 and state the main result established in this paper; cf.~Theorem~\ref{nofoliationNull}.
In Section~\ref{sec:NullConj}, we derive a lower bound on the null conjugacy radius of a point,
by analyzing Jacobi fields along null geodesics
and using, first, an affine parameter and, next, the time parameter of the foliation.
Our main estimate about the injectivity radius is proven in Sections~\ref{sec:intersecting}
and \ref{sec:convex}. Finally, in Section~\ref{examp}, we exhibit a class of spacetimes satisfying
all the bounds assumed in our main theorem but no lower curvature bound.


\section{Terminology and main result}
\label{obs}

\subsection{Lorentzian manifolds endowed with a foliation}

Let $(\M, \g)$ be a time-oriented, $(n+1)$-dimen\-sional Lorentzian manifold (without boundary).
We write $\la X, Y \ra = \la X, Y \ra_{\g} := \g(X, Y)$ for the scalar product of two vectors $X, Y$
with respect to the metric $\g$.
Throughout, we fix a point $p \in M$
and assume the existence of a set $M_I \subseteq \M$ containing $p$ that is
globally hyperbolic and geodesically complete and, therefore,
foliated by the level hypersurfaces of a time function
$t \colon M_I \to I$. The latter is normalized to be monotonically increasing towards the future
and onto an interval $I \subseteq \R$.
Moreover, we denote the level sets of this function by
$\Hcal_t$ ($t \in I$) with $0 \in I$ and $p \in \Hcal_0$
and, when convenient, we will assume the normalization $I=[-1,0]$.

From the exterior derivative of the function $t$, we define the
{\bf lapse function} of the foliation, $n \colon M_I \to (0, +\infty)$ by the relation
$$
n := (-\g(dt, dt))^{-1/2},
$$
where $\g(dt, dt)$ is the $(2, 0)$-version of the metric acting on the
one-form $dt$. (Notationally, we do not distinguish here between the $(0, 2)$- and
$(2, 0)$-versions of the metric, denoting both by $\g$.) We define the
future-directed {\bf unit normal vector field} $\T$ on $M_I$
to be the unique vector field on $M_I$ determined by the relation
\[
\g(\T, \cdot) := - n dt.
\]

Since $M_I$ is geodesically complete and globally hyperbolic,
the hypersurfaces $\Hcal_t$ are necessarily diffeomorphic (cf.~Geroch~\cite{Geroch})
and, more specifically, we may introduce a diffeomorphism
$$
\phi_t \colon \Hcal_0 \to \Hcal_t, \qquad t < 0,
$$
whose inverse is determined by transporting any point $q \in \Hcal_t$ along the integral curve of $\T$ through $q$
to its point of intersection with $\Hcal_0$.
The manifold
$\Hcal_0$ inherits a one-parameter family of Riemannian metrics $\g_t :=
\phi_t^* \big( \g|_{\Hcal_t} \big)$ (with $t \in I$), where $\g|_{\Hcal_t}$ is the induced Riemannian metric
on the hypersurface $\Hcal_t$.

To state our other geometrical bounds, we require some additional
 objects associated with the foliation. First, given $p \in M_I$ and the normal vector field $\T$ associated with the foliation,
we define the {\bf reference Riemannian metric} on $M_I$ by
\[
\gT = \g + 2 \, \g(\T, \cdot) \otimes \g(\T, \cdot).
\]
The metric $\gT$ can be used to define inner products and norms on tensor
bundles on $M_I$, which we denote by $\IPT{\cdot}{\cdot}$ and
$\normT{\cdot}$, respectively.
An additional geometrical object characterizing the foliation is the
{\bf deformation tensor}, which is an element of the space of all symmetric $(0, 2)$-tensor fields
on $M_I$, defined by $\pibf := \mathscr{L}_{\T} g$ where $\mathscr{L}$ is the Lie derivative operator.

While our statements are fully geometric, it will be convenient to give proofs in local coordinates.
Given any local coordinates $\{ x^i \}$ on a subset $U \subseteq
\Hcal_0$,
we may define a transported coordinate system $(t, x^i)$ on the set $I \times U \subseteq M_I$
by translating points in $\Hcal_0$ along the integral curves of the vector field
$\T$. In terms of these local coordinates, the Lorentzian
metric takes the form
\bel{foliation}
\g = - n^2 \, dt^2 + g_{ij} dx^i dx^j,
\ee
where $n$ is the lapse function. Then, $\g_t(x) := \phi_t^* \left( g_{ij}(t, \phi_t(x)) dx^i dx^j \right)$
($x \in \Hcal_0$)
defines the relevant metric on $\Hcal_0$ in local coordinates. In terms of this transported coordinate system, we have
\[
\T = \frac{1}{n} \partial_t,
\]
and the reference Riemannian metric takes the local form
\[
\gT = n^2 dt^2 + g_{ij} dx^i dx^j.
\]
Finally, the deformation tensor, $\pibf$ has components
\be
\label{assume-pi}
\pibf_{tt} = 0, \qquad
\pibf_{ti} = \partial_i n, \qquad
\pibf_{ij} = \frac{1}{n} \partial_t g_{ij}
\ee
with respect to the transported coordinate system.


\subsection{Geometric bounds}
\label{sec:bounds}

We are now in a position to state our assumptions on the foliation (and the spacetime).
On the initial slice, we assume a lower bound $\iota_0>0$ on the {\bf initial injectivity radius}
of the manifold $(\Hcal_0, \g_0)$ at the point $p$:
\begin{flalign}
\label{inj-initial2}
&\text{\bf Condition } (\iota_0) :
\qquad
\begin{aligned}
& \Inj(p, \Hcal_0, g_0) \geq \iota_0.
\end{aligned} &
\end{flalign}
Recall that, when the manifold $\Hcal_0$ is closed (i.e. compact without boundary),
a theorem by Cheeger~\cite{Cheeger:finiteness} (see also~\cite{HeintzeKarcher}) provides an estimate for the
injectivity radius of $\Hcal_0$ in terms of its diameter and volume and
an (upper and lower) bound on its sectional curvature.
When $\Hcal_0$ is non-compact, then Cheeger, Gromov, and Taylor's theorem~\cite{CheegerGromovTaylor}
provides an injectivity radius estimate under an upper bound on the sectional curvature
and a lower bound on the volume of metric balls at $p$.
s
Concerning the lapse function, we assume
that there exists a positive constant $\Kbar_n$, referred to as the {\bf upper lapse constant},
with the property:
\begin{flalign}
\label{nbound}
&\mathbf{Condition} \ (\Kbar_n):
\quad
n \le \Kbar_n \qquad \text{ in the set } M_I.&
\end{flalign}
Importantly, we do not require a lower bound on the lapse, a fact that may be of interest in applications.
(See, for instance, \cite{CMCSchwarzschild}.)

It is convenient to state the remaining assumptions directly in the initial slice, as we now explain.
Considering the {\bf past null cone} $\Ncal^-(p)$ from $p$
and, for $t \in I$, we introduce the intersection of this cone with the slice $\Hcal_t$:
$$
\mathcal{S}_t := \Ncal^-(p) \cap \Hcal_t.
$$
Considering the null cone $N^-(p) \subset T_p M$ in the tangent space at $p$
and all radial null geodesics from $p$,
we denote by
$$
\Sigma_t \subset N^-(p) \subset T_p M
$$
the set of points whose image (via the exponential map) lies in $\mathcal{S}_t$,
and we denote their union by
$$
N_t^-(p) := \bigcup_{t \leq s \leq 0} \Sigma_s.
$$
Observe that, for sufficiently small values of $t$ at least, the sets $\Sigma_t$ define a foliation of $N^-_t(p)$.
Finally, we define the image of the set $\mathcal{S}_t$ in the slice $\Hcal_0$ by
$$
S_t := \phi_t^{-1}(\mathcal{S}_t) \subset \Hcal_0.
$$
For sufficiently small $t < 0$ at least,
the set $S_t$ is topologically a sphere of dimension $(n-1)$, and the set
$$
\mathscr{N}^-_t(p) := \bigcup_{t \le s \le 0} S_t \subset \Hcal_0
$$
is (topologically) a closed $n$-dimensional ball. This set
$\mathscr{N}^-_t(p)$ ---the main object of study in our analysis---
is thus obtained by projecting on $\Hcal_0$ the null cone $\mathscr{N}^-(p)$ (a subset of the spacetime).

In addition, let $\Fcal^-_t(p)$ be the family of all
(restrictions of) $t$-parametrized radial geodesics $\gamma \colon [0,1] \to \Ncal^-_t(p)$
that originate at $p$ from a null tangent vector in $N^-_t(p)$.
We also introduce the bundle $T\Fcal^-_t(p)$ of all tangent vectors to geodesics
in $\Fcal^-_t(p)$.

Our main assumption is that the curvature operator is bounded above along
null geodesics.
Given a causal vector $X$,
we denote by $\Rm_X \colon \big\{X\big\}^\perp \to \big\{X\big\}^\perp$ the curvature operator, regarded as
the linear map
\[
\la \Rm_X (Y), Y \ra_g := \la \Rm(X,Y) Y, X \ra_g, \qquad Y \in \big\{X\big\}^\perp.
\]
Specifically, we impose that there exists a real constant $\Kbar_\Rm$,
called the {\bf upper null curvature constant of } $p$,
with the property:
\begin{flalign}
\label{nullcurvaturecondition}
&\text{\bf Condition } (\Kbar_\Rm):
\quad
\begin{aligned}
& \la \Rm_X(Y), Y \ra_g \leq \Kbar_\Rm \, \la Y, Y\ra_g,
\qquad  X \in T\Fcal^-_t(p), \, Y \in \big\{ X \big\}^\perp.
\end{aligned} &
\end{flalign}

In addition, there exists a constant $K_\pibf \ge 0$, the {\bf first deformation constant,}
such that  $\pibf$ is bounded with respect to the metric $\gT$:
\begin{flalign}
\label{pibound+}
&\mathbf{Condition}\ (K_\pibf):
\quad
\aligned
& |\pibf(V, V)| \le K_\pibf \langle V, V \rangle_{\T},
\qquad
 V \in T_qM, \quad q \in M_I.
\endaligned
&
\end{flalign}

The final bound that we require relates to the properties of the null geodesics when projected to the manifold $\Hcal_0$
and is now stated in terms of covariant derivative operators.
Letting $\nabla^{\g_t}$ and $\nabla$ be the Levi-Civita connections of the metrics $\g_t$ and $\g_0$,
respectively,
our final geometrical bound is about the {\bf second deformation tensor,}
$\omegabf$, defined as the difference between these Levi-Civita connections:
\be
\label{defome}
\omegabf(X, Y) := \nabla_X^{\g_t} Y - \nabla_X Y, \qquad X, Y \text{ vector fields on } \Hcal_0.
\ee
We assume that there exists a constant $K_\omegabf$, referred to as the {\bf second deformation constant,}
with the property:
\begin{flalign}
\label{radialacceleration}
&\mathbf{Condition} \ (K_\omegabf):
\quad
\aligned
& |\omegabf(V, V)| \le K_\omegabf \langle V, V \rangle_{\T},
\qquad
 V \text{ spacelike or null, in the set } M_I.
\endaligned
&
\end{flalign}

To conclude this section, we emphasize that there is no assumption that our
constants are {\sl small}, merely that the geometrical quantities mentioned above are uniformly bounded.
In particular, we have not assumed that our metric is in any sense
{\sl close} to the (flat) Minkowski metric, or that the curvature of the metric $\g$ is small.
For instance, if the metric, and its connection, are changing rapidly along the $t$-foliation
(as in the \lq\lq bump\rq\rq\ example mentioned in Example~\ref{bump1}, below), the values of
 the constants would be large.


\subsection{Null injectivity estimate}

As noted above, a given local foliation of the spacetime leads to a foliation of (a subset of) the null cone,
whose geometry is now investigated.

\begin{definition}
Given a point $p$ and a local foliation $\Hcal_t$, the {\bf past null injectivity radius} of $p$
(with respect to this foliation) is
denoted by
$$
\NullInj^-(p),
$$
and
is the supremum of all values $|t|$ such that the exponential map $\exp_p$
is a global diffeomorphism from the pointed null cone $N^-_t(p) \setminus \big\{ 0\big\}$ in $T_p M$
to its image in the manifold.
\end{definition}

\begin{theorem}[Null injectivity estimate]
\label{nofoliationNull}
Fix any positive constants $\iota_0$,
$\Kbar_n$,
$\Kbar_\Rm$, $K_\pibf$, and $K_\omegabf$.
Let $(\M, \g, p)$ be a time-oriented, pointed Lorentzian manifold such that,
along some foliation defined in a subset $M_I$ containing $p$,
the conditions $(\iota_0)$, $(K_n)$, $(\Kbar_\Rm)$, $(K_\pibf)$, and $(K_\omegabf)$
are satisfied.
Then there exists a real $\iota>0$,
depending only on (the dimension $n$ and) the constants above,
with the property that the past null injectivity radius of $p$ is bounded below by $\iota$, that is,
$$
\NullInj^-(p) \geq \iota.
$$
\end{theorem}

In earlier works, injectivity radius estimates were established
under an $L^2$ curvature bound (Klainerman and Rodnianski \cite{KR4})
or under a
sup-norm bound on the curvature (Chen and LeFloch \cite{ChenLeFloch}). The above theorem encompasses
 spacetimes not covered in these works and
 for which {\sl no uniform lower bound} on the curvature may not be available.

Note that $(\Kbar_\Rm)$ is the only condition involving second-order derivatives of the metric, while
the remaining conditions involve zero- or first-order derivatives.
Most importantly, in Theorem~\ref{nofoliationNull}, we do not assume a lower bound on the
curvature nor on the lapse function.


To establish the above theorem, we must derive an estimate for the largest value $|t|$,
denoted by $|t_1|$,
such that for all $|t| < |t_1|$,
the set $N^-_t(p) \setminus \{ 0 \}$ is globally diffeomorphic to its image via the exponential map.
It follows from~\cite[Theorem~9.15]{BEE} that the
null exponential map, $\left. \exp_p \right| N^-_{t_1}(p)$, breaks down as a global
diffeomorphism if and only if one (or both) of the following possibilities occur:
\begin{itemize}

\item[$\bullet$] There exists a point $q \in \Hcal_{t_1}$ that is conjugate to $p$
along a null geodesic from $p$ to $q$.

\item[$\bullet$] There exists $q \in \Hcal_{t_1}$ such that there exist distinct null
geodesics from $p$ that intersect at the point $q$.
\end{itemize}
This is quite similar to the situation on a complete Riemannian
manifold, where a result of Whitehead states that a geodesic, $\gamma$, from a point $p$ ceases
to be minimizing at
a point $q$ if and only if either $q$ is conjugate to $p$ along $\gamma$ and/or
there exists a distinct geodesic from $p$ to $q$ of the same length as $\gamma$.

Let us recall that in Riemannian geometry, there are therefore two key ingredients involved in
proving injectivity radius estimates (see, for instance, \cite{Cheeger:finiteness}). First, one derives an estimate that
gives a lower bound on the conjugacy radius, i.e.~the distance that one must
travel along a radial geodesic from a point before one encounters a conjugate
point. Such an estimate is usually found by a Rauch comparison argument, and
requires an {\sl upper bound} on the sectional curvature along the geodesics.

The second ingredient required for an injectivity radius estimate is a lower bound
on the length of the shortest geodesic loop through a point in the manifold (or
the shortest closed geodesic, in the case that the manifold is compact). Such an
estimate generally requires different geometrical conditions. For
example, Cheeger's lower bound on the length of the shortest closed geodesic on
a compact manifold~\cite{Cheeger:finiteness} (cf. also~\cite{HeintzeKarcher})
requires a {\sl lower bound} on the sectional curvature and volume of the manifold and
an upper bound on the diameter.

We tackle the problem of determining an estimate for the null injectivity radius
on a Lorentzian manifold in a similar way.
In the next section, we consider conjugate points along null geodesics.
A lower bound on the null
conjugacy radius for affinely-parametrized null geodesics is obtained under
upper bound on the curvature along null geodesics (following here~\cite{Harris}).
We then translate this result in terms of the $t$-foliation. It is at this point that
our assumed bounds on the lapse and second fundamental form
of the $t$-foliation are required. The second issue, that of intersecting null geodesics,
is treated in Sections~\ref{sec:intersecting} and~\ref{sec:convex} in which the strategy of proof is presented.


\section{Null conjugacy radius}
\label{sec:NullConj}

\subsection{Estimate based on the foliation parameter}

We begin with a definition and our main result in this section.

\begin{definition}
Given a point $p$ and a local foliation $\Hcal_t$ containing $p$,
the {\bf past null conjugacy radius} of $p$ (with respect to this foliation) is denoted by
\[
\NullConj^-(p),
\]
and is the supremum of all values $|t|$ for which the restriction of the exponential map $\exp_p \colon N^-_t(p) \to M$
is a local diffeomorphism onto its image.
\end{definition}

\begin{proposition}[Conjugacy radius estimate based on the foliation parameter]
\label{tconjugacyradius}
Consider the null cone $\Ncal^-(p)$ from a point $p$ and assume that
Conditions $(\Kbar_n)$, $(\Kbar_\Rm)$, and $(K_\pibf)$
on the lapse, curvature operator, and deformation tensor, respectively (i.e.~\eqref{nbound},
\eqref{nullcurvaturecondition},
and~\eqref{pibound+}) are satisfied.
If $\Kbar_\Rm \leq 0$, then no null geodesics from $p$ have conjugate points.
If $\Kbar_\Rm > 0$, then there
exists a real $\iota = \iota(\Kbar_\Rm, \Kbar_n, K_\pibf) > 0$ such that no null geodesics
from $p$ have conjugate points for $t$ larger than $\iota$ and, specifically,
$$
\NullConj^-(p, T_p) \geq
\iota
: =
\begin{cases}
\frac{1}{\Kbar_n K_\pibf} \log \left( \frac{\sqrt{\Kbar_\Rm}}{\sqrt{\Kbar_\Rm} + \pi K_\pibf} \right),
 &K_\pibf > 0,
\\
- \frac{1}{\Kbar_n} \frac{\pi}{\sqrt{\Kbar_{Rm}}},
 &K_\pibf = 0.
\end{cases}
$$
\end{proposition}

To establish this result, we are going to derive, from the assumed curvature bound,
a corresponding bound on the curvature along affinely-parametrized geodesics.
In Section~\ref{sec:Harris}, below, we will then follow Harris~\cite{Harris} and estimate the corresponding conjugacy radius
along affinely-parametrized geodesics.  Finally, in Section~\ref{sec:tconj},
we then translate this bound into an estimate for the conjugacy radius with respect to the $t$-foliation.


\subsection{Estimate based on the affine parameter}
\label{sec:Harris}

We begin with the following result.

\begin{proposition}[Conjugacy radius estimate based on the affine parameter]
\label{lemma:Harris}

Let $\gamma$ be an affinely-parametrized, past directed null geodesic from the point $p$.
Let $K$ be a constant such that, along the geodesic $\gamma$, the curvature operator with respect to $\gamma' := \frac{d\gamma}{ds}$ is bounded above by $K$, i.e.
\[
\Rm \left( \gamma', X, \gamma', X \right) \le K \la X, X \ra, \qquad X \in \left\{ \gamma' \right\}^{\perp}.
\]
If $K \le 0$, then for all $s > 0$ the point $\gamma(s)$ is not conjugate to $p$ along $\gamma$.
If $K > 0$, then the point $\gamma(s)$ is not conjugate to $p$ along $\gamma$ for $s < \frac{\pi}{\sqrt{K}}$, at least.
\end{proposition}

In order to prove this proposition,
 we recall that conjugate points are determined by the differential of the exponential map. This map is non-degenerate at $V \in T_p M$ precisely when all non-trivial Jacobi fields, $\J$, along the geodesic $\gamma_V$ with $\J(0) = 0$ are non-vanishing at the point $\exp_p V$. The Jacobi equation is essentially a system of second-order differential equations, and the behavior of its solutions may be controlled by comparison with solutions of model differential equations.

Let $\gamma$ be a past-oriented, affinely-parametrized, null geodesic from $p$. We consider an arbitrary Jacobi field $\J$ along $\gamma$, satisfying, by definition, the Jacobi equation
$$
\aligned
& \J^{\prime\prime}(s) + \Rm \big( \J(s), \gamma'(s) \big) \gamma'(s) = 0,
\\
& \J(0)=0, \qquad \J'(0) \neq 0,
\endaligned
$$
where $\J' := \nabla_{\gamma'} \J$, etc. It follows directly from the Jacobi equation and the condition that $\J(0) = 0$ that
\[
\la \gamma'(s), \J(s) \ra = \la \gamma'(0), \J'(0) \ra s.
\]
Therefore, if $\la \gamma'(0), \J'(0) \ra \neq 0$, then $\la \gamma'(s), \J(s) \ra \neq 0$ for $s > 0$ and hence
${\J(s) \ne 0}$ for $s > 0$. Since such a Jacobi field cannot give rise to a conjugate point along $\gamma$
and without loss of generality from the point of view of detecting conjugate points,
we restrict our attention to Jacobi fields that are orthogonal to $\gamma'$.

Since $\gamma$ is a null geodesic, the condition that $\J \perp \gamma'$ allows
 that $\J$ may have a component parallel to $\gamma'$. As noted in~\cite[pp.~562]{BE1} a Jacobi field parallel
to $\gamma$ leads to the index form along
$\gamma$ having a degeneracy since $I[\J, V] = 0$, for all $V \in V_{\perp}(c)$. As such, the link between non-definiteness of the index form and the existence of conjugate points is lost. There are two distinct, but essentially equivalent, ways to dealing with this issue:
\begin{itemize}

\item[a).] Uhlenbeck~\cite{U}, Beem and Ehrlich~\cite{BE1}, and Hawking and Ellis~\cite{HE} consider equivalence classes of Jacobi fields where $\J_1 \sim \J_2$ if $\J_1 - \J_2$ is a multiple of $\gamma'$.

\item[b).] Harris~\cite{Harris} impose that Jacobi fields along $\gamma$ are \lq\lq nowhere tangential\rq\rq\ to $\gamma$.

\end{itemize}
For our purposes, it is more convenient to follow the second approach. Recall that Harris defines a (perpendicular) Jacobi field $\J$ along a null geodesic $\gamma$ to be {\sl nowhere tangential} to $\gamma$ if, for any $s$ such that $\J(s) \neq 0$, then $\J(s)$ is not proportional to $\gamma'(s)$. He defines a Jacobi field to be {\sl purely tangential} if
the proportionality condition
$\J(s) \propto \gamma'(s)$ holds for all $s$. It is then straightforward to prove,
from the uniqueness theorems for second-order ordinary differential equations, that up to the first conjugate point along $\gamma$, a perpendicular Jacobi field is either purely tangential or nowhere tangential.
The Jacobi equation implies that any purely tangential Jacobi field is of the form $\J(s) = A s \gamma'(s)$, where $A$ is a constant, and hence will be non-zero for $s > 0$. As such, purely tangential Jacobi fields do not give rise to conjugate points along the geodesic $\gamma$. We are therefore finally lead to restrict ourselves to
{\sl non-tangential,
perpendicular Jacobi fields} along $\gamma$. Such a Jacobi field may be written in the form
\be
\label{expressionJacobi}
\J(s) = \alpha(s) \gamma'(s) + X(s),
\ee
where $X(s)$ is a space-like vector field along $\gamma$ that is orthogonal to $\gamma'$
and $\alpha$ vanishes whenever $X$ vanishes. We then note that
\[
\la \J(s), \J(s) \ra = \la X(s), X(s) \ra \ge 0,
\]
with equality if and only if $\J(s) = 0$ (since $\J$ is assumed nowhere tangential). As such, there is an induced Euclidean inner product induced on the space of Jacobi fields under consideration and, therefore, we write
\[
|\J(s)| := \sqrt{\la \J(s), \J(s) \ra}.
\]

Motivated by the study of Jacobi fields on Riemannian manifolds with constant curvature,
it is natural to introduce for each $K \in \RR$ the real-valued function
\be
\label{sn}
\Phi_K(t) :=
\begin{cases}
|K|^{-1/2} \, \sinh (|K|^{1/2} \, t),  &  K < 0,
\\
t,                                  & K = 0,
\\
K^{-1/2} \, \sin (K^{1/2} \, t),    & K > 0,
\end{cases}
\ee
which we define for all $t\geq 0$ if $K \leq 0$, and for $t \in [0, \pi \, K^{-1/2}]$ if $K > 0$.
Observe that $\Phi_K$ thus defined is non-negative (on its domain of definition), and satisfies
\be
\label{sn34}
\ddot{\Phi}_K + K \, \Phi_K = 0, \qquad
\Phi_K(0) = 0, \quad
\Phid_K(0) = 1.
\ee

The following lemma provides us the desired estimate for the length of Jacobi fields, and
Proposition~\ref{lemma:Harris} follows immediately from this lemma .
The proof given below is adapted from arguments in Riemannian geometry, and
a (more general) version of this result, together with a different proof, appeared in Harris~\cite{Harris}.

\begin{lemma}[Jacobi field estimate]
\label{conjlemma2}
Let $\gamma$ be a past-directed, affinely-parametrized,
null geodesic satisfying $\gamma(0) = p$.
Let {\J} be a nowhere tangential, perpendicular Jacobi field along $\gamma$ with $\J(0) = 0$
(and, in particular, $\J(s) \perp \gamma'(s)$).
Then, $K \in \RR$ being the constant defined as in Proposition~\ref{lemma:Harris}, the Jacobi field {\J} satisfies
\be
|\J(s)| \geq | \J'(0) | \, \Phi_{K}(s), \qquad s > 0.
\label{Rauch}
\ee
In particular, if $K \leq 0$ then for any $s \in [0, 1]$ the point $\gamma(s)$ is not conjugate to $p$ along $\gamma$. If $K > 0$ then $\gamma(s)$ is not conjugate to $p$ along $\gamma$ for any $s < \pi/\sqrt{K}$.
\end{lemma}

\begin{proof}
We first check that
$$
\lim_{s \to 0} \frac{d}{ds} |\J(s)| = | \J'(0) |.
$$
Namely, with $\varphi(s) := \norm{\J(s)}$ we can write
$$
\aligned
\lim_{s \to 0} (\varphi'(s))^2
& = \lim_{s \to 0} {\la Y(s), \J'(s) \ra^2 \over \la Y(s), Y(s) \ra}
\\
& = \lim_{s \to 0} {2 \, \la Y(s), \J'(s) \ra
\over 2 \, \la Y(s), \J'(s) \ra} \, \Big( \la \J'(s), \J'(s) \ra^2
     - \Rm(Y(s), \gamma'(s), Y(s), \gamma'(s) ) \Big)
\\
& = | \J'(0) |^2,
\endaligned
$$
where we have used L'H\^opital's rule and Jacobi field's equation.

In addition, we have
$$
\frac{1}{2} \frac{d}{ds} |Y(s)|^2 = \g(Y(s), \J'(s)),
$$
and an application of the Cauchy-Schwarz inequality yields the Kato-type inequality
\be
\label{KATO}
\left\vert \frac{d}{ds} | Y(s)| \right\vert \leq |\J'(s)|
\ee
for all $s$ such that $\J(s) \neq 0$.

We now calculate
$$
\aligned
\frac{1}{2} \frac{d^2}{ds^2} \norm{\J(s)}^2
& =\frac{d}{ds} \g(\J(s), \J'(s))
\\
&= \norm{\J'(s)}^2 + \g(\J(s), \J^{\prime\prime}(s)),
\endaligned
$$
thus
$$
\aligned
\frac{1}{2} \frac{d^2}{ds^2} \norm{\J(s)}^2
&= \norm{\J'(s)}^2 - \Rm \big( \J(s), \gamma'(s), \J(s), \gamma'(s)\big).
\endaligned
$$
Therefore, imposing our curvature assumption, we deduce that
$$
\aligned
\frac{1}{2} \frac{d^2}{ds^2} \norm{\J(s)}^2
&\ge \norm{\J'(s)}^2 - K \, \norm{\J(s)}^2.
\endaligned
$$

As above, let $\varphi(s) := \norm{\J(s)}$,
which has the properties that $\varphi(0) = 0$,
$\varphi'(0) = |\J'(0)|$,
which is non-vanishing (and positive) since $\J' (0) \neq 0$.
Applying~\eqref{KATO}, we then deduce that
$$
{\varphi}^{\prime\prime}(s) + K \, \varphi(s) \ge 0
$$
at all points where $\varphi(s) \neq 0$. Using the function $\Phi$ defined in equation~\eqref{sn}, we now define the function
$$
\psi(s) := \varphi'(0) \Phi_{K}(s),
$$
which satisfies
$$
{\psi}^{\prime\prime}(s) + K \, \psi(s) = 0
$$
with initial conditions $\psi(0) = \varphi(0) = 0$, $\psi'(0) = \varphi'(0)$. We then deduce that
\[
\frac{d}{ds} \left( \varphi'(s) \psi(s) - \varphi(s) \psi'(s) \right) \ge 0, \qquad s > 0,
\]
so $\varphi'(s) \psi(s) - \varphi(s) \psi'(s)$ is non-decreasing for $s > 0$. The initial conditions therefore imply that
$\varphi'(s) \psi(s) - \varphi(s) \psi'(s) \ge 0$ for $s > 0$. Therefore, we have
\[
\frac{d}{ds} \left( \frac{\varphi(s)}{\psi(s)} \right) \ge 0,
\]
so the ratio $\varphi(s)/\psi(s)$ is non-decreasing. An application of L'H\^opital's rule implies that
the ratio $\varphi(s)/\psi(s)$ tends to $1$ at the origin, giving
\[
\varphi(s) \ge \psi(s), \qquad t \ge 0.
\]
Rewriting this inequality in terms of the Jacobi field $\J$ and the functions $\Phi$,
we arrive at the required inequality~\eqref{Rauch}.
In particular, we note that the first zero of $\varphi$ cannot occur prior to the first zero of $\psi$,
\end{proof}


\subsection{Derivation of the main conjugate radius estimate}
\label{sec:tconj}

To derive the conjugacy radius estimate with respect to the $t$-foliation,
we require a more detailed analysis of the
geodesic equation. Let $\gamma \colon [0, a] \to M$ be an
affinely-parametrized, past-directed, null
geodesic emanating from the point $p$. In terms of the local transported coordinate
description~\eqref{foliation}, writing the components of the geodesic in the
form $s \mapsto (t(s), x^i(s))$, the equations for affinely-parametrized
geodesics
with respect to the metric {\g} take the form
\begin{subequations}
\begin{align}
\frac{d^2 t}{ds^2} + \frac{n_t}{n} \left( \frac{dt}{ds} \right)^2 +
2 \frac{n_i}{n} \left( \frac{dt}{ds} \right) \left( \frac{dx^i}{ds} \right)
+ \frac{1}{2n^2} \left( \partial_t g_{ij} \right) \frac{dx^i}{ds}
\frac{dx^j}{ds}&= 0,
\\
\frac{d^2 x^i}{ds^2} + n g^{ij} n_j \left( \frac{dt}{ds} \right)^2 +
g^{ij} \left( \partial_t g_{jk} \right) \left( \frac{dt}{ds} \right) \left(
\frac{dx^k}{ds} \right)
+ \Gamma^i{}_{jk} \frac{dx^j}{ds} \frac{dx^k}{ds}&= 0,
\end{align}
\label{s:geodesics}\end{subequations}
where the coefficients $\Gamma^i{}_{jk}$ are Christoffel symbols.
We wish to consider such null geodesics parametrized by the foliation parameter
$t$. As such, we view a null geodesic as a map $I \ni t \mapsto
(t, x(t)) \in M_I \simeq I \times \Hcal_0$, where
(in the notation of Section~2.1)
$x(t) := \phi_t^{-1}(\gamma(t)) \in \Hcal_0$ is the spatial projection of the
geodesic. In terms of the local coordinates above, we now consider the affine
parameter, $s$, and the components $x^i$ of the geodesic as functions of $t$.
Re-arranging the above equations, we find that $s = s(t)$
and $x^i = x^i(t)$ satisfy the equations
\begin{subequations}
\begin{gather}
\frac{\ddot{s}}{\dot{s}} = \frac{n_t}{n} + 2 \frac{n_i}{n} \dot{x}^i +
\frac{1}{2n^2} \left( \partial_t g_{ij} \right) \dot{x}^i \dot{x}^j,
\label{t:geodesics}
\\
\ddot{x}^i + n g^{ij} n_j + g^{ij} \left( \partial_t g_{jk} \right) \dot{x}^k +
\Gamma^i{}_{jk} \dot{x}^j \dot{x}^k = \frac{\ddot{s}}{\dot{s}} \dot{x}^i,
\label{xi:geodesics}\end{gather}\end{subequations}
where a dot $\dot{}$ denotes $\tfrac{d}{dt}$, and a dash ${}^{\prime}$ denotes $\tfrac{d}{ds}$.

We wish to consider geodesics, parametrized by $t$, that are, in addition, null.
From~\eqref{foliation}, we deduce that such geodesics have the
additional property
\[
n(t, x(t))^2 = g_{ij}(t, x(t)) \frac{dx^i(t)}{dt} \frac{dx^j(t)}{dt}.
\]

\begin{lemma}
\label{lemma:pi}
Along past-directed geodesics from $p$, one has
\bel{ndsdt}
\frac{d}{dt} \log \left( \frac{1}{n} \left| \frac{ds}{dt} \right| \right) = \frac{1}{2n}
\pibf \left( \frac{d\gamma}{dt}, \frac{d\gamma}{dt} \right)
\ee
and, using the affine parameter,
\bel{ndtds}
\frac{d}{ds} \left( n \frac{dt}{ds} \right) = - \frac{1}{2}
\pibf \left( \frac{d\gamma}{ds}, \frac{d\gamma}{ds} \right).
\ee
\end{lemma}

\begin{proof}
The result is local, so we may carry out the calculations in the adapted coordinate system described above.
The future-directed, unit normal $\T$ takes the form $\frac{1}{n} \partial_t$,
whereas the $t$-parametrized tangent vector takes the form
\[
\dot{\gamma} = \frac{d\gamma}{dt} = \partial_t + \frac{dx^i}{dt}
\partial_{x^i},
\]
and we therefore have
\[
\langle \T, \dot{\gamma} \rangle = - n.
\]
If $s$ is the affine parameter along the geodesic $\gamma$, and we denote
$\tfrac{d}{ds}$ by ${}^{\prime}$, then
\[
\gamma' = \frac{d\gamma}{ds}
= \frac{d\gamma}{dt} \left(
\frac{ds}{dt} \right)^{-1}
= \frac{\dot{\gamma}}{\dot{s}}.
\]
We therefore have
\[
\nabla_{\dot{\gamma}} \dot{\gamma} = \frac{\ddot{s}}{\dot{s}} \dot{\gamma},
\]
and thus
\begin{align*}
\pibf(\dot{\gamma}, \dot{\gamma}) &=
\left( \mathscr{L}_{\T} \g \right) \left( \dot{\gamma}, \dot{\gamma} \right) =
2 \langle \nabla_{\dot{\gamma}} \T, \dot{\gamma} \rangle
= 2 \Big( \nabla_{\dot{\gamma}} \langle \T, \dot{\gamma} \rangle - \langle \T,
\nabla_{\dot{\gamma}} \dot{\gamma} \rangle \Big)
\\
&= 2 \, \Big( - \nabla_{\dot{\gamma}} n - \frac{\ddot{s}}{\dot{s}} \langle
\T, \dot{\gamma} \rangle \Big)
= 2 \Big( - \frac{dn}{dt} + n \frac{\ddot{s}}{\dot{s}} \Big)
\\
&= 2 n \frac{d}{dt} \log \left( \frac{1}{n} \left| \frac{ds}{dt} \right| \right).
\end{align*}
(Note that, since we are considering past-directed null geodesics, we have $\tfrac{ds}{dt} < 0$.)

For the second result, we simply observe that
\begin{align*}
\pibf(\gamma', \gamma') &=
2 \langle \nabla_{\gamma'} \T, \gamma' \rangle
= 2 \nabla_{\gamma'} \langle \T, \gamma' \rangle
= 2 \frac{d}{ds} \left( - n \frac{dt}{ds} \right).
\end{align*}
\end{proof}

Lemma~\ref{lemma:pi} above allows us to translate between the behavior of geodesics in terms of
the affine parameter $s$, and the time parameter $t$. We will also require the
following observation.

\begin{lemma}
\label{lemma:normalisation}
After a suitable normalization of the affine parameter,
one has
\label{lemma:dsdt}
\[
\left. \frac{ds}{dt} \right|_{t=0} = - n(p).
\]
\end{lemma}

\begin{proof}
Given a past-directed, null vector $L \in T_p \M$ such that $\g(\T, L) = +1$,
then $s$ is the affine parameter along the geodesic $\gamma_{L} \colon [0,
s_{L}] \to M$ uniquely determined by the condition that $\gamma_{L}(0) = p$,
$\dot{\gamma}_{L}(0) = L$. Therefore, from the definition of $\T$, we obtain
\[
n(p) \left. \frac{dt}{ds} \right|_{t=0} =
\langle n(p) dt, \dot{\gamma}_{L}(0) \rangle =
- \langle \T(p), L \rangle = - 1.
\]
\end{proof}

\begin{lemma}
\label{prop:stinequality}
Under the conditions~$(\Kbar_n)$ and $(K_\pibf)$ imposed along past-directed, null
geodesics from $p$, one has
\bel{stderivt}
- e^{\Kbar_n K_\pibf t} \le \frac{1}{n} \frac{ds}{dt} \le - e^{- \Kbar_n K_\pibf t}
\ee
and, in terms of the affine parameter,
\bel{stderiv}
- 1 - K_\pibf s \le \frac{1}{n} \frac{ds}{dt} \le - 1 + K_\pibf s.
\ee
The affine parameter and $t$-parameter along a past-directed null geodesic satisfy the inequality
\bel{stinequality}
s(t) \le \frac{1}{K_\pibf} \left( 1 - e^{K_\pibf \Kbar_n t} \right).
\ee
\end{lemma}

\begin{proof}
From condition~\eqref{pibound+}, we find that
\[
\left| \pibf \left( \frac{d\gamma}{dt}, \frac{d\gamma}{dt} \right) \right| \le
K_\pibf \left| \frac{d\gamma}{dt} \right|_{\T}^2.
\]
Moreover, we have
\[
\left| \frac{d\gamma}{dt} \right|_{\T}^2 = n^2 + \left\vert \dot{x}(t) \right\vert_{\g_t}^2
= 2 n^2,
\]
since the geodesic $\gamma$ is assumed null and, therefore, from~\eqref{ndsdt}
\[
\left| \frac{d}{dt} \log \left( \frac{1}{n} \left| \frac{ds}{dt} \right| \right) \right|
\le K_\pibf n \le K_\pibf \Kbar_n.
\]
Integrating this inequality from $t$ to $0$ and recalling that $t < 0$, we find
\[
- K_\pibf \Kbar_n t \le
- \log \left( \frac{1}{n} \left| \frac{ds}{dt} \right| \right)
\le K_\pibf \Kbar_n t,
\]
and re-arranging gives~\eqref{stderivt}. The first inequality in~\eqref{stderivt} then yields
\[
\frac{ds}{dt} \ge - n e^{\Kbar_n K_\pibf t} \ge - \Kbar_n e^{\Kbar_n K_\pibf t}.
\]
Integrating this inequality from $t$ to $0$ then gives~\eqref{stinequality}.

In order to deduce the inequality~\eqref{stderiv}, we note that
\[
\left| \pibf \left( \frac{d\gamma}{ds}, \frac{d\gamma}{ds} \right) \right| \le
K_\pibf \left| \frac{d\gamma}{ds} \right|_{\T}^2 =
2 K_\pibf n^2 \left( \frac{dt}{ds} \right)^2.
\]
Therefore, from~\eqref{ndsdt}, we deduce that
\[
\left| \frac{d}{ds} \left( n \frac{dt}{ds} \right) \right|
\le K_\pibf \left( n \frac{dt}{ds} \right)^2
\]
and, hence
\[
\left| \frac{d}{ds} \left( n \frac{dt}{ds} \right)^{-1} \right| \le K_\pibf.
\]
Integrating, with the boundary condition as in Lemma~\ref{lemma:normalisation}, gives the inequalities~\eqref{stderiv}.
\end{proof}

\begin{proof}[Proof of Proposition~\ref{tconjugacyradius}]
Let $\gamma$ be a past-directed, null geodesic from the point $p$ parametrized by the foliation parameter $t$. Our bound~\eqref{nullcurvaturecondition} on the curvature operator implies that, for all $Y \perp \gammad$, we have
\[
R \left( \frac{d\gamma}{dt}, Y, \frac{d\gamma}{dt}, Y \right) \le \Kbar_{Rm} n^2 \la Y, Y \ra.
\]
Changing to affine parametrization of the geodesic $\gamma$, we therefore find that
\[
R \left( \frac{d\gamma}{ds}, Y, \frac{d\gamma}{ds}, Y \right) \le
\Kbar_{Rm} \left( n \frac{dt}{ds} \right)^2 \la Y, Y \ra.
\]
From the second inequality in~\eqref{stderiv}, we therefore deduce that
\[
R \left( \frac{d\gamma}{ds}, Y, \frac{d\gamma}{ds}, Y \right) \le
\frac{\Kbar_{Rm}}{\left( -1 + K_\pibf s \right)^2} \la Y, Y \ra.
\]
For the moment, we assume that $K_\pibf > 0$. Let $s_0 = s_0(K_\pibf) := \tfrac{1}{K_\pibf}$, and assume that $0 \le s \le s_1$ for fixed $s_1 < s_0$. For $s \in [0, s_1]$, the curvature operator for the affinely-parametrized geodesic $\gamma$ satisfies
\[
\Rm_{\frac{d\gamma}{ds}} \le
\frac{\Kbar_{Rm}}{\left( -1 + K_\pibf s_1 \right)^2}.
\]
If follows, from Proposition~\ref{lemma:Harris}, that for $\Kbar_{Rm} \le 0$, the geodesic $\gamma$
will contain no conjugate points for $s \le s_1$. If $\Kbar_{Rm} > 0$, there will be no conjugate points for
$s^2 < \pi^2/{K\big(\Kbar_{Rm}, K_\pibf, s_1\big)}$, where
\[
K(\Kbar_{Rm}, K_\pibf, s_1) := \frac{\Kbar_{Rm}}{\left( -1 + K_\pibf s_1 \right)^2}.
\]
If $\pi^2/{K(\Kbar_{Rm}, K_\pibf, s_1)} > (s_1)^2$, then no conjugate points occur for $s < s_1$. We may therefore repeat our conjugate point estimate with a larger value of $s_1$. Alternatively, if $\pi^2/{K(\Kbar_{Rm}, K_\pibf, s_1)} < (s_1)^2$,
then a conjugate point occurs for some $s < s_1$. In this case, we should repeat our conjugate point calculations with a smaller value of $s_1$. The optimal estimate is therefore achieved if we solve for $s_1$ such that $\pi^2/{K(\Kbar_{Rm}, K_\pibf, s_1)} = (s_1)^2$. This yields the estimate
\bel{s1}
s_1 = \frac{\pi}{\sqrt{\Kbar_{Rm}}} \left( 1 + K_\pibf \frac{\pi}{\sqrt{\Kbar_{Rm}}} \right)^{-1}.
\ee
Note that $s_1 \le s_0$ since $K_\pibf > 0$.

As such, our estimates show that a past-directed, affinely-parametrized null geodesic from $p$
will encounter no conjugate points for $s < s_1$, with $s_1 := s_1(\Kbar_{Rm}, K_\pibf)$ defined as in equation~\eqref{s1}. We must now translate this condition for the geodesic parametrized by the foliation parameter $t$. Any conjugate point must occur for a value of $s$ greater than or equal to $s_1$. It follows from~\eqref{stinequality} that this occurs at a value of $t$, $t_1$, such that
\[
s_1 \le \frac{1}{K_\pibf} \left( 1 - e^{K_\pibf \Kbar_n t_1} \right).
\]
Re-arranging this expression yields the estimate stated in the proposition.

Finally, if $K_\pibf = 0$, then $\Rm_{\frac{d\gamma}{ds}} \le \Kbar_{Rm}$. Again, if $\Kbar_{Rm} \le 0$, there are no conjugate points and if $\Kbar_{Rm} > 0$ there are no conjugate points prior to $s_0 = \pi/\sqrt{\Kbar_{Rm}}$. Applying the inequality~\eqref{stinequality} in the limiting case $K_\pibf \to 0$ then yields
\[
t_1 \le - \frac{1}{\Kbar_n} \frac{\pi}{\sqrt{\Kbar_{Rm}}},
\]
which completes the proof of the proposition.
\end{proof}


\section{Geodesic intersections far from the vertex}
\label{sec:intersecting}

\subsection{Strategy of proof}

We now concentrate on the case where past-directed null geodesics from $p$ reintersect at a
point $q \in \Hcal_{t_0}$ for some $t_0 \in I$ with $t_0 < 0$.
We assume that the values of $t$ that we consider
are sufficiently small that there are no null conjugate points, and that the
null exponential map is
therefore a local diffeomorphism. The breakdown of the null exponential map as a
global diffeomorphism at $t_0$ implies that we have distinct null geodesics
from $p$, $\gamma_1$ and $\gamma_2$ (which we take to be parametrized by the
parameter $t$) such that $\gamma_1(t_0) = \gamma_2(t_0) =: q \in \Hcal_{t_0}$,
and that this phenomenon does not happen for any $t > t_0$.
By construction, the tangent vectors $\dot{\gamma}_1(0), \dot{\gamma}_2(0) \in T_p \M$ are distinct
null vectors at $p$. Following \cite[Lemma~3.1]{KR4}, we prove the following result
which we observe to be valid for arbitrary metrics.

\begin{lemma}[Projection of intersecting null geodesics]
\label{KR-lemma}
The spatial projections with respect to the $t$-foliation of the null tangent vectors $\dot{\gamma}_1(t_0), \dot{\gamma}_2(t_0) \in T_q M$ are opposite.
\end{lemma}

\begin{proof}
We first translate this information into our picture on the manifold $\Hcal_0$.
The intersection of null geodesics from $p$ on the hypersurface $\Hcal_{t_0}$ implies that the sphere
$S_{t_0}$ has a self-intersection at the point
$q_0 := \phi_{t_0}^{-1}(\gamma_1(t_0)) = \phi_{t_0}^{-1}(\gamma_2(t_0))$.
The definition of $t_0$ implies that the spheres $S_t$, for $t_0 < t < 0$ have no self-intersection. We consider the
projections
$$
x_1 := \phi_t^{-1} \circ \gamma_1 \colon [t_0, 0] \to \Hcal_0, \qquad
x_2 := \phi_t^{-1} \circ \gamma_2 \colon [t_0, 0] \to \Hcal_0.
$$
We wish to prove that $\dot{x}_1(t_0) \propto - \dot{x}_2(t_0)$, where the constant of proportionality is positive. Since we have no conjugate points and the spheres $S_t$ do not self-intersect for $t_0 < t$, the $S_t$ are embedded spheres in $\Hcal_0$ while $S_{t_0}$ is immersed but not embedded. As is geometrically clear (and technically follows from a transversality argument~\cite{KR4}) the self-intersection of the sphere $S_{t_0}$ must be
tangential. Since the normal vectors (with respect to the metric $\g_t$ on $\Hcal_0$) to $S_{t_0}$ at the point of intersection are the tangent vectors $\dot{x}_1(t_0), \dot{x}_2(t_0)$, it follows that these vectors must be proportional. The constant of proportionality cannot be positive, since uniqueness of solutions of the geodesic equations would then imply that $\gamma_1 = \gamma_2$. Moreover, since $x_1, x_2$ are
projections of null geodesics, it follows that $\dot{x}_1$ and $\dot{x}_2$ must be non-vanishing. Hence the constant of proportionality cannot be zero, and therefore must be negative.
\end{proof}

We wish to study the minimal value of $t$, denoted by $t_0$, for which the sphere $S_t$
self-intersects. Recall that
 $\injec(\g_0, \Hcal_0,p)=:r_0$ denotes the injectivity radius at the point $p \in \Hcal_0$ with respect to the metric $\g_0$. Then there are two possibilities:
\begin{itemize}

\item[$\bullet$] {\bf Geodesic intersections far from $p$.}
A point in $S_t$ leaves the ball $B_{\g_0}(p, r_0)$ at or
before time $t_0$, i.e. there exists a null geodesic from $p$ with the property
that its projection $\Gamma=\Gamma(t)$ satisfies
$d_{\g_0}(p, \Gamma(t)) \ge r_0$ for some $t \le t_0$;

\item[$\bullet$] {\bf Geodesic intersections near $p$.}
$S_t$ self-intersects before any point in $S_t$ reaches distance
$r_0$ from $p$.
\end{itemize}

We will study the first possibility in the next subsection and, in this case,
Proposition~\ref{prop:t0bound}, below, gives a
lower bound on $t_0$. The second possibility is more involved, and is the subject
of Section~\ref{sec:convex}.


\subsection{Geodesic intersections far from $p$}

We first note that we may recover the second fundamental form, $\mathbf{k}_t$, of the hypersurface $\Hcal_t$ from the spatial projection of the deformation tensor $\pibf$. In particular, if $X$ is a vector field on $\Hcal_t$ (and, therefore, $X \perp \T$), then we have
\be
\label{defpi}
\pibf(X, X) = \left( \mathscr{L}_{\T} \g \right)(X, X) = 2 \la \nabla_X \T, X \ra
=: 2 \mathbf{k}_t(X, X).
\ee
It therefore follows that if the deformation tensor satisfies condition~$(K_\pibf)$, then we have a corresponding bound on the second fundamental form:
\be
\label{333}
|\mathbf{k}_t(X, X)| = \frac{1}{2} |\pibf(X, X)| \le \frac{1}{2} K_\pibf \gT(X, X) = \frac{1}{2} K_\pibf \g_t(X, X).
\ee

The following is the main result of this subsection.

\begin{proposition}
\label{prop:t0bound}
Suppose that the foliation satisfy the conditions~$(\iota_0)$,
$(\Kbar_n)$, and $(K_\pibf)$.
If $S_t$ does not intersect the cut locus of $p \in \Hcal_0$, that is,
$S_t \cap \mathrm{Cut}_{\g_0}(p) = \emptyset$ then one has
\[
|t| \geq
\begin{cases}
\frac{2}{\Kbar_n K_\pibf} \log \left( 1 + \frac{1}{2} K_\pibf \iota_0 \right),
& K_\pibf \neq 0,
\\
\frac{\iota_0}{\Kbar_n}, & K_\pibf = 0.
\end{cases}
\]
\end{proposition}

Hence, Proposition~\ref{prop:t0bound} provides us with a lower bound on the value of $t$ for which the sphere $S_t$ leaves $B_{\g_0}(p, \iota_0)$. To establish this proposition,
 we require some estimates for length of the spatial
projection of a null geodesic.

\begin{lemma}
For $-1 \le s \le t \le 0$, one has
\bel{gsgtequiv}
\g_t
\le e^{\Kbar_{n} K_\pibf (t-s) } g_s.
\ee
\end{lemma}

\begin{proof}
We first note that, in the transported coordinate system, we have
\[
\mathbf{k}_t := \frac{1}{2n} \partial_t \g_t.
\]
If $X$ is a smooth vector field on $\Hcal_0$ (independent of $t$),
then we have
\[
\aligned
\partial_t \left( \g_t(X, X) \right)
& =
2 n(t) \mathbf{k}_t(X, X)
\\
& \le
n(t) K_\pibf \g_t(X, X) \le
\Kbar_n K_\pibf \g_t(X, X).
\endaligned
\]
Integrating from $s$ to $t$, we therefore have
\[
\g_t(X, X) \le
e^{\Kbar_n K_\pibf (t-s)} \g_s(X, X).
\]
This inequality holds for all $X$, so we deduce the
inequality~\eqref{gsgtequiv}.
\end{proof}


Given a null geodesics $\gamma \colon [-1, 0] \to M$, we wish to consider the
length of its spatial projection $\Gamma$ with respect to the metric $\g_0$.

\begin{lemma}
For $-1 \le t \le 0$ one has
\bel{g0Length}
L_{\g_0}[\left.\Gamma\right|_{[0, t]}] \le
\begin{cases}
\frac{2}{K_\pibf} \left( - 1 + e^{- \frac{1}{2} \Kbar_n K_\pibf t} \right),
& K_\pibf > 0,
\\
\Kbar_n |t|,     & K_\pibf = 0.
\end{cases}
\ee
\end{lemma}

\begin{proof} We first treat the case $K_\pibf > 0$. We then have
\begin{align*}
L_{\g_0}[\left.\Gamma\right|_{[0, t]}] &=
\int_t^0 \left| \dot{\mathbf{x}}(u) \right|_{\g_0} \, du
\le \int_t^0 e^{- \frac{1}{2} \Kbar_n K_\pibf u} \left| \dot{\mathbf{x}}(u) \right|_{\g_u} \, du
\\
&= \int_0^t e^{- \frac{1}{2} \Kbar_n K_\pibf u} \, n(u, x(u)) \, du
\le \Kbar_n \int_t^0 e^{- \frac{1}{2} \Kbar_n K_\pibf u} \, du
\\
&= \frac{2}{K_\pibf} \left( - 1 + e^{- \frac{1}{2} \Kbar_n K_\pibf t} \right),
\end{align*}
as required. The result for $K_{\pibf}$ follows from taking the limit as $K_{\pibf} \to 0$ of this inequality. 
\end{proof}

\begin{corollary}
\label{cor:dg0}
Given any null geodesic  and its projection
$$
\gamma \colon [-1, 0] \to M_I, \qquad
\Gamma: [-1, 0] \to \Hcal_0,
$$
then, for $-1 \le t \le 0$, one has
\bel{dg0}
d_{\g_0}(p, \Gamma(t)) \le
\begin{cases}
\Kbar_n |t|, & K_\pibf = 0,
\\
\frac{2}{K_\pibf} \left( - 1 + e^{- \frac{1}{2} \Kbar_n K_\pibf t} \right),
 & K_\pibf \neq 0.
\end{cases}
\ee
\end{corollary}

\begin{remark}
It follows from the above results that, if $\gamma_1, \gamma_2 \colon [t_0, 0] \to M_I$
are null geodesics that intersect at $p$ and again in the surface $\Hcal_{t_0}$ for some $t_0$,
then the lengths of their spatial projections $\Gamma_1, \Gamma_2 \colon [0,
t_0] \to \Hcal_0$ satisfy~\eqref{g0Length}. The closed loop at $p$ defined by
concatenating these curves, which we denote by $\Gamma$, therefore satisfies the
inequality
$$
L_{\g_{t_0}}[\Gamma] \le
\frac{4}{K_\pibf} \left( - 1 + e^{- \frac{1}{2} \Kbar_n K_\pibf t} \right).
$$
\end{remark}

\begin{proof}[Proof of Proposition~\ref{prop:t0bound}]
If $S_t \cap \mathrm{Cut}_{\g_0}(p) = \emptyset$ then there exists a null
geodesic $\gamma$ such that its projection, $\Gamma$, satisfies
\[
d_{\g_0}(p, \Gamma(t)) \ge r_0 \ge i_0,
\]
for some $t$. The result then follows from Corollary~\ref{cor:dg0}.
\end{proof}


\section{Geodesic intersections near the vertex}
\label{sec:convex}

\subsection{Tangent space calculations}

We now consider the case where there exists $t_0 < 0$ such that the spheres $S_t$
have no self-intersections for $t_0 < t < 0$, the sphere $S_{t_0}$ has self-intersections,
and the spheres $S_t, 0 \le t \le t_0$ are contained in the ball of center $p$ and radius $\iota_0$.

Since we remain within the injectivity radius (with respect to the metric $\g_0$) at $p$, we may use the exponential map (with respect to $\g_0$) to define spheres
\[
\Sigma_t := \left( \exp_{p}^{\g_0} \right)^{-1} (S_t) \subset T_{p} \Hcal_0.
\]
Since $\exp_{p}$ is a global diffeomorphism from $B(0, \iota_0)$ to $B(p, \iota_0)$, a sphere $S_t$
will have self-intersections if and only if the sphere $\Sigma_t$ has self-intersections.

In their approach to the null injectivity radius problem~\cite{KR4}, where the metric {\g} is shown to be
$\eps$-close to the Minkowski metric in a particular local coordinate system, Klainerman and Rodnianski
argue that the intersection of null geodesics that we are considering cannot occur within the local coordinate chart.
Intuitively, it seems clear that, in order for the light cone to become so distorted that it self-intersects in the required fashion, we would require a significant amount of curvature in our manifold, and hence the metric cannot be assumed globally close to the Minkowski metric. The following example shows, however, the phenomenon that we are considering {\sl cannot} be ruled out, in general.

\begin{example}[Growing bump metric]
\label{bump1} 
Let $\Hcal = \R^2$, with $\g_0$ the (flat) Euclidean metric.
Therefore the injectivity radius of the initial slice is
$+\infty$, so any intersections of the $S_t$ that happen will occur
before they intersect the cut locus of $p$ with respect to $\g_0$.
We evolve the metric for $t > 0$ so that it gains a bump, of height $t$
say, with $\g_t$ being the induced metric from flat $\R^3$.
If we pick a point away from the bump, then the geodesic balls
will start out round but, once they hit the bump (which is growing with $t$),
they will bend around the bump. The geodesic balls will then intersect at the
back of the bump for sufficiently large $t$. If we project to the surface
$\Hcal_0$, then the arrangement of the spheres $S_t$ is a
family of nested spheres developing a self-intersection.

Note that, although we have assumed $\Hcal_0$ to be $\R^2$, it is clear that this
argument may be suitably localised in order to make $\Hcal_0$ compact, and to any
dimension greater than or equal to $2$.
\end{example}

One of the notable features of this example is that the metric $\g_t$ is evolving
with $t$, so the spatial geometry is undergoing significant change.

Having established that a self-intersection of a sphere $S_t$ can occur,
the main result of this section is that we can find an explicit lower bound on the corresponding value of $t$.
We first must find a condition that is necessary for the intersection
of null geodesics, for which we can then develop an estimate. Although the criterion that we
will use
is not optimal, it does fulfil this requirement.

Let $\gamma \colon [0, a] \to M$ be a past-directed null geodesic with
$\gamma(0) = p$, with $\gamma(s) = (t(s), x(s))$, where
$x(s) \in \Hcal_0$ (for $s \in [0, a]$) is its projection. The equations that $t$ and $x$ must
satisfy in an arbitrary transported local coordinate system are given in
equations~\eqref{s:geodesics}.

Generically, for sufficiently negative values of $t$, there may exist
values of $t$ for which distinct null geodesics from $p$ intersect on
the hypersurface $\Hcal_t$. Let $t_* < 0$ denote the largest (i.e. least negative)
value of $t$ for which there exist distinct null geodesics from $p$,
$\gamma_1$ and $\gamma_2$, such that $\gamma_1(t_*) = \gamma_2(t_*) \in \Hcal_{t_*}$.
We denote this point of intersection by $q$. Lemma~\ref{KR-lemma}
shows that the spatial projection with respect to the $t$-foliation of the tangent vectors
$\dot{\gamma}_1(t_*)$ and $\dot{\gamma}_2(t_*)$ at $q$ are opposite. Denoting
the projections of these geodesics to $\Hcal_0$ by $x_1$ and $x_2$, we therefore
deduce that $\dot{x}_1(t_*) \propto - \dot{x}_2(t_*)$, where the constant of
proportionality is positive.
 
We define the radial function $r(x) := d_{\g_0}(p, x)$ with respect to the metric $\g_0$ on the hypersurface $\Hcal_0$. Recall (see, e.g.~\cite{Petersen}) that $r$ thus defined is a smooth distance function on the set ${\Hcal_0 \setminus \left( \{ p \} \cup \Cut(p) \right)}$. Following~\cite{Petersen}, we denote by $\partial_r = \nabla r$ the corresponding unit radial vector field, which is smooth on the ball $B_{\g_0}(p, i_0)$ away from the point $p$.

We then have the following preliminary result.

\begin{lemma}
\label{prop:scalarprods}
One of the inner products $\langle \dot{x}_1(t_*), \partial_r
\rangle_{\g_0}, \langle \dot{x}_2(t_*), \partial_r \rangle_{\g_0}$ is
{\sl non-negative}.
\end{lemma}

\begin{proof}
From the fact that $\dot{x}_1(t_*) \propto - \dot{x}_2(t_*)$, we have
\[
\langle \dot{x}_1(t_*), \partial_r \rangle_{\g_0} \propto - \langle
\dot{x}_2(t_*), \partial_r \rangle_{\g_0},
\]
where the constant of proportionality is positive.
\end{proof}

We may therefore derive an upper bound on the possible value of $t_*$
by finding an upper bound on the value of $t$ for which the spatial projection of null
geodesics satisfies $\langle \dot{x}(t), \partial_r \rangle_{\g_0} \ge 0$.
For calculational simplicity, we will work with affinely-parametrized geodesics.
Since $\frac{ds}{dt} < 0$ for past-directed null geodesics,
what we will derive is a lower bound for the first value of $s$, $s_*$, for which the tangent vector
$\frac{d\gamma}{ds}$ satisfies
\be
\label{cond777}
\left< \frac{dx}{ds}(s), \partial_r \right>_{\g_0} \le 0.
\ee
We will then turn this estimate into an estimate for the corresponding value of $t$
using the same method as we used for the conjugate point estimate.

Therefore, let $x \colon [0, a] \to \Hcal_0$ be the projection to $\Hcal_0$ of an affinely-parametrized null geodesic, with affine parameter $s$. Therefore $x(0) = p \in \Hcal_0$ and, by assumption, $x(s)$ remains within the ball
of center $p$ and radius $i_0 := \iota_0$, i.e. $x(s) \in B_{\g_0}(p, i_0)$ for $s \in [0, a]$. As such, there
exists a unique, affinely-parametrized, radial geodesic (with respect to $\g_0$)
$\gamma_s \colon [0, 1] \to \Hcal_0$, from $p$ to $x(s)$. This geodesic has the property that
\[
d_{\g_0}(p, x(s)) = L_{\g_0}[\gamma_s] = \int_0^1
\sqrt{\g_0\left(\frac{d\gamma_s(u)}{du}, \frac{d\gamma_s(u)}{du}\right)} du.
\]

Let $\nabla$ be the Levi-Civita connection associated with the metric $\g_0$ and, as before, let ${}'$ denote $\frac{d}{ds}$. We then have the following result.

\begin{lemma}
\label{prop:VOA}
One has
\begin{align*}
\frac{d}{ds} d_{\g_0}(p, x(s)) &=
\left< \partial_r, \frac{dx}{ds}(s) \right>_{\g_0},
\\
\frac{d^2}{ds^2} d_{\g_0}(p, x(s)) &\ge
\left< \partial_r, \nabla_{x'(s)} x'(s) \right>_{\g_0},
\end{align*}
and also
\[
\left. \frac{d}{ds} d_{\g_0}(p, x(s)) \right|_{s=0} = 1.
\]
\end{lemma}

\begin{proof}
The first equality follows directly from the first variation of arc-length formula.
Given a geodesic
$\gamma$ in a manifold and an orthogonal vector field, $W$, along $\gamma$, then the index form is given by
(see, e.g.,~\cite{CE})
\[
I[W] :=
\int_0^1 \left( |\nabla_{\gamma'} W|^2 -
\left< \gamma', \mathbf{R}(W, \gamma') W \right> \right) \, ds
\]
Recall that the index form along a geodesic is positive-definite prior to the first conjugate point along the geodesic.
We define the variation through geodesics
$\alpha \colon [0, 1] \times [0, a] \to \Hcal_0$ with $(u, s) \to \gamma_s(u)$,
and the tangent vectors $T := \alpha_*(\partial_u), V := \alpha_*(\partial_s)$.
The second variation of arc-length formula then states that
\[
\frac{d^2}{ds^2} d_{\g_0}(p, x(s)) =
\left[ \frac{1}{|T|} \vphantom{|^|} \langle T,
\nabla_{V} V \rangle_{\g_0} \right]^{u=1}_{u=0} +
\frac{1}{|T|} I[V^{\perp}],
\]
where $V^{\perp}$ denotes the part of the vector field $V$
along $\gamma_s$ that is orthogonal (with respect to $\g_0$) to $T$.

We now apply this formula to our situation. Since $\gamma_s(0) = p$,
for each $s$, we deduce that $V(u=0) = 0$.
$V^{\perp}$ is therefore a vector field along $\gamma_s$ that
vanishes at $p$, and is equal to $x'(s)$ at the point $x(s)$. A standard
index argument (cf.~again \cite{CE})
implies that $I[V^{\perp}] \ge I[\mathbf{J}]$,
where $\mathbf{J}$ is the Jacobi vector field along $\gamma_s$ such that
$\mathbf{J}(p) = 0$ and $\mathbf{J}(x(s)) = V^{\perp}(x(s))$. By
assumption, there are no points conjugate to $p$ along the geodesics $\gamma_s$
(since we are working within $B_{\g_0}(p, i_0)$). Therefore the index form is
positive definite along the geodesic $\gamma_s$, so $I[\mathbf{J}] > 0$. We
therefore have
\[
\frac{d^2}{ds^2} d_{\g_0}(p, x(s)) \ge
\left< \partial_r, \nabla_{x'(s)} x'(s) \right>_{\g_0},
\]
as required.

Finally, note that $x'(0)$ is an outward-directed, radial vector at the point $p$. Therefore, as $s \to 0$, the inner product $\left< \partial_r, \frac{dx}{ds}(s) \right>_{\g_0}$ converges to $\left| \frac{dx}{ds}(0) \right|_{\g_0}$. We then have
\[
\left. \left| \frac{dx}{ds} \right|_{\g_0} \right|_{s=0} =
\left. \frac{dt}{ds} \right|_{s=0}
\left. \left| \frac{dx}{dt} \right|_{\g_0} \right|_{t=0} =
\left. \left| \frac{ds}{dt} \right|^{-1} \right|_{t=0} n(p) = 1,
\]
by Lemma~\ref{lemma:dsdt}.
\end{proof}

We now impose, for some constant $\barK_{acc.}$, the following {\bf radial acceleration condition:}
\begin{flalign}
\label{radcond}
&\mathbf{Condition} (\barK_{acc.}):
 \qquad \left< \partial_r, \nabla_{x'(s)} x'(s) \right>_{\g_0} \geq \barK_{acc.} &
\end{flalign}
along the projections, $x$, of null geodesics.

\begin{lemma}
\label{sintersectionestimate}
Under the assumed conditions, any pair of past-directed, null
geodesics from $p$ will not intersect for $0 < s < s_*$, where
$$
s_* = s_*(\barK_{acc.}) =
\begin{cases} - \frac{1}{\barK_{acc.}},    &\barK_{acc.} < 0,
\\
+\infty, &\barK_{acc.} \ge 0. \end{cases}
$$
\end{lemma}

\begin{proof}
From Lemma~\ref{prop:VOA} we deduce that
\[
\frac{d}{ds} d_{\g_0}(p, x(s)) \ge 1 + \barK_{acc.} s.
\]
Therefore, $\frac{d}{ds} d_{\g_0}(p, x(s)) > 0$ for $s < s_*$, where $s_*$ is as
stated in the lemma.
Lemma~\ref{prop:scalarprods} and the following discussion then complete the proof.
\end{proof}

We are in a position to establish the main result of this section.

\begin{proposition}
\label{thm:tintersection}
Under the conditions $(\Kbar_n)$, $(K_\pibf)$, $(\barK_{acc.})$,
any pair of past-directed, null geodesics from $p$ do not intersect for $t_0 < t < 0$, where
$$
t_0 = t_0(\Kbar_n, K_{\pi}, \barK_{acc.})
:= \frac{1}{\Kbar_n K_{\pi}} \log \left( 1 + \frac{K_{\pi}}{\barK_{acc.}} \right).
$$
\end{proposition}

\begin{proof}[Proof of Proposition~\ref{thm:tintersection}]
It is now sufficient to apply the same technique as employed in the proof of Proposition~\ref{tconjugacyradius}
to translate the affine parameter estimate in Lemma~\ref{sintersectionestimate}
to an estimate for the corresponding value of $t$.
\end{proof}


\subsection{The radial acceleration condition}
\label{sec:radialacceleration}

We now discuss the nature of the condition \eqref{radcond},
where we have assumed a lower bound on the radial acceleration.
In the transported coordinate system, the $\nabla_{x'(s)} x'(s)$ term takes the form
\begin{align*}
&\left( \nabla_{x'(s)} x'(s) \right)^i
\\
&= \frac{d^2}{ds^2} x^i(s) + \Gamma^i{}_{jk}(0, x(s)) \frac{dx^j}{ds}
\frac{dx^k}{ds}
\\
&= \left( \Gamma^i{}_{jk}(0, x(s)) - \Gamma^i{}_{jk}(t(s), x(s)) \right)
\frac{dx^j}{ds} \frac{dx^k}{ds}
- n \frac{dt}{ds} g^{ij} \left( \pibf_{jt} \frac{dt}{ds} + \pibf_{jk}
\frac{dx^k}{ds} \right),
\end{align*}
where we have used the fact that $x(s)$ is the projection of a null geodesic in the manifold $M$.
Recall, from Section~\ref{sec:bounds}, that we have defined the $(1, 2)$ tensor field $\omegabf$ on $\Hcal_0$
as the difference between the Levi-Civita connections of the metrics $\g_t$ and $\g_0$.
It then follows that
\[
\langle \partial_r, \nabla_{x'(s)} x'(s) \rangle_{\g_0} =
- \alphabf(t) \left( \frac{d\gamma}{ds}, \frac{d\gamma}{ds} \right),
\]
where
\bel{Aequation}
\aligned
& \alphabf(t) \left( V, V \right)
\\
& :=
\langle \partial_r, n \grad_{\g_t} n \rangle_{\g_0} \left( V^t \right)^2 +
2 n V^t \langle \partial_r, \mathbf{k}(V^{\perp}) \rangle_{\g_0}
+ \langle \partial_r, \omegabf(V^{\perp}, V^{\perp}) \rangle_{\g_0}.
\endaligned
\ee

Our radial acceleration condition is therefore equivalent to a bound of the form
\[
- \alphabf(t) \left( \gamma', \gamma' \right) \ge \barK_{acc.}
\]
along the null geodesics $\gamma$. It is clear, from the explicit form of $\alphabf$ given in~\eqref{Aequation}, that such a bound follows from our assumptions $(K_\pibf)$ and $(K_\omegabf)$.

The tensor $\omegabf$ itself may be estimated by calculating the difference between the Christoffel symbols of the metric $\g_t$ and $\g_0$.
\[
\omegabf^i{}_{jk} (t, x(t)) = \Gamma^i{}_{jk}(t, x(t)) - \Gamma^i{}_{jk}(0, x(t)) =
\int_0^t \left( \partial_1 \Gamma^i{}_{jk} \right)(u, x(t)) \, du.
\]
The derivative in the integrand may be written in the form
\begin{align*}
\partial_t \Gamma^i{}_{jk}
&= R_{tj}{}^i{}_k + \frac{1}{2} \nabla_j \left( n \pibf^i{}_k \right) +
\frac{1}{2} \left( \pibf_{tk} \pibf^i{}_j - \pibf_t{}^i \pibf_{jk} \right).
\end{align*}
As such, the $\omegabf$ term could be bounded by assuming additional bounds on the $R_{ti}{}^j{}_k$ parts of the curvature, on the spatial derivative of the $n \pibf_{ij}$ (i.e. the spatial derivative of $\partial_t g_{ij}$) and additional bounds on the deformation tensor.

It is clear that the radial acceleration condition,
or a condition of a similar type, is required here.
The crux of our argument is that the projection of the null geodesics in $M$ to the hypersurface $\Hcal_0$ are not geodesics on $\Hcal_0$. In order to estimate the deformation of the spheres, $S_t$, we need to control how much these projections deviate from geodesics with respect to $\g_0$. The radial acceleration (or equivalently, the form $\alphabf$) is the most direct way of measuring this deviation.


\section{A class of spacetimes with curvature bounded above}
\label{examp}

The conditions assumed in our main theorem in this paper are satisfied by a large class of spacetimes.

\begin{proposition}[Family of spacetimes with curvature unbounded below]
Fix some positive constants $\iota_0$,
$\Kbar_n$,
$\Kbar_\Rm$, $K_\pibf$, and $K_\omegabf$.
There exists
a family of spacetimes satisfying all of the assumptions in Theorem~\ref{nofoliationNull},
but whose curvature operator is not uniformly bounded below in terms of the given constants,
even in the $L^2$ norm.
\end{proposition}

\begin{proof} We search for the desired spacetimes in the class of plane wave solutions.
The standard four-dimensional plane wave metric takes the form
\[
\g = - 2 \, du dv + 2 H(u, x, y) du^2 + dx^2 + dy^2,
\]
where $u,v,x,y$ are local coordinates.
With respect to the null coframe
\[
e^1 = du, \qquad
e^2 = dv - H(u, x, y) du, \qquad
e^3 = dx, \qquad
e^4 = dy,
\]
in terms of which
\[
\g = - e^1 \otimes e^2
- e^2 \otimes e^1
+ e^3 \otimes e^3
+ e^4 \otimes e^4,
\]
we find that the non-vanishing components of the curvature tensor are
\be
\label{curv1}
R_{1313} = - H_{xx}, \qquad R_{1314} = - H_{xy}, \qquad R_{1414} = - H_{yy}.
\ee

If we wish to consider the curvature operator, then we must consider the curvature quantity $R(X, Y, X, Y)$, where $X$ is a null vector and $Y$ is orthogonal to $X$.
Letting $\Sigma^{ij} := X^i Y^j - X^j Y^i$ then, from the explicit formula for the curvature components above, we find that
\bel{planewavecurvature}
R(X, Y, X, Y) =
- H_{xx} \left( \Sigma^{13} \right)^2
- 2 H_{xy} \Sigma^{13} \Sigma^{14}
- H_{yy} \left( \Sigma^{14} \right)^2.
\ee
In particular, if we impose that the Hessian (in the variables $x, y$) of the function $H$ is positive semi-definite,
\bel{Hessian}
H_{xx} \geq 0,
\qquad
H_{xx} H_{yy} - H_{xy}^2 \ge 0,
\ee
then we deduce that the curvature operator $\Rm_{X}$ is non-positive
\[
\Rm_{X} \le 0.
\]
The additional geometrical conditions required for our theorems require only the first derivatives of the metric.
 As such, these conditions can be satisfied by choosing $H$ to have bounded first derivatives.

It is clear, however, from~\eqref{planewavecurvature} that we can make the curvature arbitrarily negative
 by letting (for example) $H_{xx}$ and $H_{yy}$ become arbitrarily large and positive.

A special case of this construction occurs if we take the function $H$ to be of the special form $H(x, y, u) = a(x, u) + b(y, u)$. The Hessian condition~\eqref{Hessian} is satisfied if $a_{xx} \ge 0$ and $b_{yy} \ge 0$, i.e. the functions $a$ and $b$ are convex in $x$ and $y$, respectively. We may choose $a$ and $b$ to have small first derivatives (in order to satisfy the conditions of Theorem~\ref{nofoliationNull}), but such that there exist points at which $a_{xx}$ and $b_{yy}$ are large and positive. In particular, we may consider a sequence of such metrics
where the limiting functions $a, b$ are convex but not $C^2$ so that the curvature of the spacetime
approaches a distribution containing a Dirac-mass singularity, say on the hypersurface $x=x_0$ for some $x_0$.
Solutions with distributional curvature \cite{LM} could be used
here to handle the spacetimes with low regularity obtained in the limit.
\end{proof}


\section*{Acknowledgments}

The first author (JDEG) was supported by START-project Y237--N13 of the Austrian Science Fund, and
is grateful to the Fakult\"{a}t f\"{u}r Mathematik, Universit\"{a}t Wien for their continuing hospitality.
The second author (PLF) was partially supported
by the Agence Nationale de la Recherche (ANR) through the grant 06-2-134423 entitled
``Mathematical Methods in General Relativity''.


\end{document}